\documentclass[12pt]{amsart}
\usepackage{amssymb}
\usepackage{amsmath}
\usepackage{bbm}
\usepackage{color}
\usepackage{verbatim}

\definecolor{darkred}{rgb}{0.75,0,0.3}

\newcommand\AND{\quad\text{and}\quad}

\newcommand\dde{\boldsymbol{\delta}}
\newcommand\dg{\mathsf{deg}^+}
\newcommand\DD{\mathsf{D}}
\newcommand\ep{\varepsilon}
\newcommand\Ex{\mathbb{E}}
\newcommand\HH{\mathsf{H}}
\newcommand\interior{\mathsf{int}\,}
\newcommand\N{\mathbb N}
\newcommand\Prob{\mathbb{P}}
\newcommand\Q{\mathbb{Q}}
\newcommand\R{\mathbb R}

\newcommand\wh{\widehat}
\newcommand\wt{\widetilde}

\numberwithin{equation}{section}

\newtheoremstyle{mythm}
  {9pt}
  {9pt}
  {\itshape}
  {0pt}
  {\bfseries}
  {}
  { }
  {\thmnumber{(#2)}\thmname{ #1}\thmnote{ #3}}

\newtheoremstyle{mydef}
  {9pt}
  {9pt}
  {\normalfont}
  {0pt}
  {\bfseries}
  {}
  { }
  {\thmnumber{(#2)}\thmname{ #1}\thmnote{ #3}}

\theoremstyle{mythm}
\newtheorem{thm}[equation]{Theorem.}
\newtheorem{pro}[equation]{Proposition.}
\newtheorem{lem}[equation]{Lemma.}

\theoremstyle{mydef}
\newtheorem{dfn}[equation]{Definition.}
\newtheorem{exa}[equation]{Example.}

\newtheorem{app}[equation]{Application.}

\newtheorem{rmk}[equation]{Remark.}

\begin{document}$\,$ \vspace{-1truecm}
\title{Comparing entropy rates on finite and infinite rooted trees} 
\author{\bf Thomas Hirschler, Wolfgang Woess}
\address{\parbox{.8\linewidth}{Institut f\"ur Diskrete Mathematik,\\ 
Technische Universit\"at Graz,\\
Steyrergasse 30, A-8010 Graz, Austria\\}}
\email{thirschler@tugraz.at, woess@tugraz.at}
\begin{abstract}
We consider stochastic processes with (or without) memory 
whose evolution is encoded by a finite or infinite rooted tree.  
The main goal is to compare the entropy rates of a given base process and
a second one, to be considered as a perturbation of the former. The processes
are described by probability measures on the boundary of the given tree, and
by corresponding forward transition probabilities at the inner nodes. The comparison is
in terms of Kullback-Leibler divergence.
We elaborate and extend ideas and results of B\"ocherer and Amjad.
Our extensions involve
length functions on the edges of the tree as well as nodes with countably many successors.
In particular, in the last part, we consider trees with infinite geodesic rays and 
random perturbations of a given process.
\\[4pt]
\end{abstract}

\maketitle

\markboth{{\sf T. Hirschler and W. Woess}}
{{\sf Entropy rates on rooted trees}}
\baselineskip 15pt

\section{Introduction}\label{sec:intro}

In \cite{BoeAm}, {\sc B\"ocherer and Amjad} consider probability
distributions on the set of leaves of a finite, rooted tree.
They derive a bound on the difference of the entropies between two such 
distributions, where the second one is ``memoryless'' in terms of
an encoding of the tree in terms of strings over a finite alphabet.
The bound is in terms of the Kullback-Leibler divergence of the first 
with respect to the second distribution, and of the lengths (heights)
of the leaves; see \eqref{eq:BoeAm} below.

In this note, we generalise these inequalities, thereby also streamlining
the mathematical substance. Our generalisations -- which are qualitative rather than
quantitative -- consist in (1) comparing
the entropies of two arbitrary distributions on the leaves in terms of
the Kullback-Leibler divergence, (2) involving more general length functions
than those where each edge of the tree has length $1$, and (3) allowing internal
nodes with a countable infinity of forward neighbours. Finally, (4),  
we also consider trees which have infinite geodesic rays, corresponding
to boundary points at infinity, and (5) consider the entropy rate of a random
perturbation of a given stochastic process.

\smallskip

In the main parts of this paper -- except for \S \ref{sec:infinite} --  
the boundary $\partial T$ of our tree $T$ with root $o$
consists of leaves (vertices with no forward neighbour) and is at most countably
infinite. It is equipped with two probability measures $\Prob$ and $\Q$, where
$\Q$ plays the role of a fixed reference measure and $\Prob$ should be thought
of as a perturbation of $\Q$ which corresponds to some experimental situation:
for example, the perturbation errors present in $\Prob$ may be inherent in the 
simulation of a stochastic process whose distribution is $\Q$.

The internal vertices (nodes) $x$ of $T$ are equipped with a function $\ell(x) > 0$
with the interpretation that each outgoing edge at $x$ has length $\ell(x)$.
The absolute value $|x|_{\ell}$ of any node is then the sum of the lengths
of the edges on the unique geodesic path from $o$ to $x$. The habitual choice
is $\ell = \ell_{\sharp} \equiv 1$. Thus we have the expected length 
$\ell(\Prob) = \sum_{v \in \partial T} |v|_{\ell}\,\Prob(v)$, and analogously 
$\ell(\Q)$. Let $\HH(\Prob)$ and $\HH(\Q)$ be the entropies of the two measures,
and $\DD(\Prob\|\Q)$ the Kullback-Leibler divergence of $\Prob$ with respect to $\Q$.
The aim is to estimate the difference between the entropy rates $\HH(\Prob)/\ell(\Prob)$
and $\HH(\Q)/\ell(\Q)$ in terms of $\DD(\Prob\|\Q)/\ell(\Prob)$.

The probability measure $\Prob$ gives rise to a forward Markov chain on the internal nodes of $T$,
whose transition probabilities are $p(y|x) = \Prob(\partial T_y)/\Prob(\partial T_x)$,
when $y$ is a forward neighbour of $x$, and where $\partial T_x$ denotes the set of leaves
$v$ for which $x$ lies on the geodesic path from $o$ to $v$. The analogous transition
probabilities associated with $\Q$ are denoted $q(y|x)$. A basic tool for achieving
the desired estimate is the ``Leaf-Average Node-Sum Interchange Theorem'' (LANSIT) of
{\sc Rueppel and Massey}~\cite{RuMa}. It relates $\HH(\Prob)/\ell(\Prob)$ with the entropies 
$\HH\bigl(p(\cdot|x)\bigr)$ and $\DD(\Prob\|\Q)/\ell(\Prob)$ with the divergences
$\DD\bigl(p(\cdot|x)\|q(\cdot|x)\bigr)$.

The main concern of \cite{BoeAm} regards the situation when the tree is finite, 
$\ell = \ell_{\sharp}\,$, and all probability measures $q(\cdot|x)=q(\cdot)$ coincide under 
a suitable labelling of the sets of forward neighbours of the internal nodes $x$. 
In this case, the LANSIT yields $\HH(\Prob)/\ell_{\sharp}(\Q) = \HH(q)$, and
the estimate of \cite{BoeAm} can be read as
\begin{equation}\label{eq:BoeAm}
\left|\frac{\HH(\Prob)}{\ell_{\sharp}(\Prob)} - \HH(q)\right| <  M\varphi(\delta) 
+ \frac{C}{\delta} \sqrt{\frac{\DD(\Prob\|\Q)}{\ell_{\sharp}(\Prob)}}\,,
\end{equation}
where $\delta$ is an arbitrarily small positive constant, $M$ is the (constant)
number of forward neighbours of the internal nodes, $\varphi(\delta) = \delta\, \log_2(1/\delta)$,
and $C \approx 1.25$.

We start in \S \ref{sec:local} by explaining our notation, which comprises some variations with respect
to previous work \cite{RuMa}, \cite{BoeAm}. We review the proof of the LANSIT 
in presence of of a general length function as above. As mentioned, our tree is allowed
to have nodes with infinitely many forward neighbours. Thus, our entropies and 
divergences can be infinite series, and we require the latter to be absolutely convergent.
Furthermore, we have local entropies $\HH\bigl(p(\cdot|x)\bigr)$ and $\HH\bigl(q(\cdot|x)\bigr)$
at the possibly infinitely many internal nodes, and we require the defining series
to be uniformly summable. See \S \ref{sec:entropies}, which consists mainly in a ``guided tour'' through
the needed summability properties and simpler conditions under which they hold.
Some further basics are also presented in \S \ref{sec:entropies}. 

Theorem \ref{thm:compare} in \S \ref{sec:compare} contains our generalisation of \eqref{eq:BoeAm}.
We consider arbitrary probability measures $\Prob$ and $\Q$ that are supported by
$\partial T$ and satisfy the mentioned summability conditions. Consequently,
our estimate of $|\HH(\Prob)/\ell(\Prob) - \HH(\Q)/\ell(\Q)|$ is a sum of three
terms, where the first can be made arbitrarily small, the central one 
comprises $\sqrt{\DD(\Prob\|\Q)/\ell(\Prob)}$, and the last one is a multiple
of the variational difference between the node-average probability measures 
induced by
$\Prob$ and $\Q$
on the set of internal vertices. The pre-factor of this term vanishes precisely
when $\HH\bigl(q(\cdot|x)\bigr)\big/\ell(x)$ is constant (or at least
asymptotically constant) on the internal nodes.
After a discussion of various aspects and variations of this estimate, in \S \ref{sec:infinite}
we turn our attention to trees which have infinite geodesic paths. In this case,
the boundary of the tree consists of the space of ``ends'' (limit points of infinite
paths) and possibly also leaves. Applying Theorem \ref{thm:compare} yields
estimates for the entropy rate and variants thereof, see Theorem \ref{thm:asy}.
In particular, coming back to the natural length function $\ell_{\sharp}\,$,
if $\Q$ is the probability distribution on the trajectory space of a stochastic process
with asymptotically constant, finite forward entropies $\HH\bigl(q(\cdot|x)\bigr)$
on the associated infinite tree and $\Prob$ is the probability governing a perturbation
of that process, then finiteness of $\DD(\Prob|\Q)$ implies that the perturbed process has the
same entropy rate as the non-perturbed one. Note again that we allow the
state space of these processes to be countably infinite, as long as $\HH\bigl(q(\cdot|x)\bigr)$
is finite. Moreover, this also holds when $\DD(\Prob_n|\Q_n)/n \to 0$,
where $\Prob_n$ and $\Q_n$ are the distributions of the respective processes up
to time $n$. A final application is given in Theorem \ref{thm:randper}, where we consider
random perturbations of a given stochastic process.

The proofs of all theorems are contained in the Appendix.

We remark that in a somewhat different vein, there has been work on the statistics of
ergodic Markov chains. {\sc Ciuperca, Girardin and Lhote}~\cite{CiGi}, \cite{CiGiLh},
as well as {\sc Yari and Nikooravesh}~\cite{YaNi}, consider the approximation 
of the entropy rate via the natural estimator (in terms of relative frequencies 
of transitions between two states, resp. visits in one state) of the chain at time 
$n$, as $n \to \infty$. Constantness of local entropies of the chain is not needed here.
Note that this is not comparison with one perturbed chain,
but with a sequence of random perturbations: the estimated transition probabilities
and invariant distribution are updated at each time step. In \cite{CiGi}, the state
space is finite, and in \cite{YaNi}, at each state there are only finitely many
positive outgoing transition probabilities. Both correspond to situations where
the corresponding infinite tree is locally finite. In  \cite{CiGiLh}, this finite
range assumption is replaced by a quasi-power property.

\section{Setting, notation, and preliminary facts}\label{sec:local}

\noindent\textbf{A. Trees, leaves, and lengths}

\smallskip

A \emph{rooted tree} is a connected graph $T$ without circles, with one vertex $o$
desginated as the root. We shall tacitly identify $T$ with the set of vertices of
the tree, and write $E(T)$ for the set of edges, if needed.

A \emph{geodesic} or \emph{geodesic path} in $T$ is
a sequence $\pi=[x_0\,,x_1\,,\dots, x_n]$ of successive neighbours in $T$ 
such that all $x_i$ are distinct. For any pair of vertices $x, y$ there is precisely one 
geodesic $\pi(x,y)$ from $x$ to $y$. Analogously, an infinite geodesic has
the form $\pi=[x_0\,,x_1\,,x_2\,,\dots]$ where $x_{k-1}$ and $x_k$ are neighbours
for each $k \in \N$.

While we  do not necessarily assume $T$ to be finite, \emph{we shall mostly assume that 
it contains no infinite geodesic,} with the exception of \S \ref{sec:infinite}.

Every vertex $x \in T \setminus \{o\}$ has a unique \emph{predecessor} $x^-$,
the neighbour of $x$ on the geodesic $\pi(o,x)\,$. In this case, we say that $x$ is
a \emph{successor} of $x^-$, and write $N(x)$ for the set of successors of $x \in T$. 
The \emph{forward degree} $\dg(x)$ of $x \in T$ 
is the number of successors of $x$. In our setting, it will be natural to assume that
$\dg(x) \ne 1$ for all $x \in T$.
The set of \emph{leaves} or \emph{boundary}
$\partial T$ of $T$ consists of all vertices which have no successor. 
By our assumption, it is non-empty, and every vertex lies on a geodesic from
$o$ to some leaf. The \emph{interior} of $T$ is $\interior T = T \setminus \partial T$.

We can describe each edge of $T$ as $e= [x^-,x]$, where $x \in T \setminus \{o\}$. We
now assign a \emph{length} $\ell(e) > 0$ to each edge $e$. 

\medskip

\noindent
$\bullet\;$ \emph{We restrict attention to length functions where $\ell(e)$
depends only on $x^-$ for  $e= [x^-,x]$.} 

\medskip

That is, at each vertex the outgoing edges have the same length, and for $e$ as
above we can write $\ell(e)=\ell(x^-)$. We remark here that this restriction 
is not completely satisfactory: in general, each edge might have a label from some
alphabet, and one would like the cost (length) of the edge to depend on that label.
However, in this case, one of the tools, the LANSIT lemma \ref{lem:lansit} fails
in general. But with the interpretation that the cost of an edge in the next step 
along $T$ depends on the \emph{last} used symbol of the alphabet, we do have a reasonable
model. For example, it is realistic in a Markovian setting of the evolution of
a language that the transitions and their cost in the next step depend on the 
last input symbol. 

Back to our length function, the length of a path
is defined as the sum of the lengths of its edges. 
For a vertex $x$,
its length (distance to $o$) $|x|_{\ell}$ is the length of $\pi(o,x)\,$.
The standard length function is induced by $\ell_{\sharp}(e) = 1$. The associated path length
is the number of edges. For $x \in T$, we then just write $|x|$ for the resulting length
of $\pi(o,x)$. This is the ordinary \emph{height} of $x$.

\medskip

\noindent\textbf{B. Forward Markov chains and leaf distributions}

\smallskip

On a rooted tree, a \emph{forward Markov chain} is given by transition
probabilities 
$$
p(y|x)\,, \quad x \in \interior T\,,\; y \in N(x)\,, \quad \text{with}
\sum_{y\in N(x)} p(y|x) = 1\,.
$$
We can interpret this via a random process: a particle starts
at $o$. If at some time it is at a vertex $x \in \interior T$ then the particle moves to a 
random successor of $x$ chosen according to the probability distribution $p(\cdot|x)$.
When the particle reaches a leaf, it stops there (so that we might add the trivial
transition probabilities $p(v|v)=1$ for $v \in \partial T$).

We now define  a measure (function) on $\partial T$, as follows.
\begin{equation}\label{eq:P} 
\Prob(v)
= \prod_{x \in \pi(o,v) \setminus \{o\}} p(x|x^-)\,, \quad v \in \partial T\,.
\end{equation}

This is the \emph{hitting distribution} on that set: for a leaf $v$, we have that 
$\Prob(v)$ is the probability that our forward Markov chain terminates at $v$. 
Conversely, given any probability measure $\Prob$ on $\partial T$, it induces 
a forward Markov chain  as follows.

\begin{dfn}\label{def:cone} 
For $x \in T$, the \emph{cone} or \emph{branch} at $x$ is the sub-tree with root $x$ given
by
$$
T_x = \{ y \in T : x \in \pi(o,y) \}
$$
\end{dfn} 

Its boundary $\partial T_x$ consists of all $v \in \partial T$ with $x \in \pi(o,v)\,$.
Then we define 
$$
p(x|x^-) = \Prob(\partial T_x)/\Prob(\partial T_{x^-})\,,
$$
where $\Prob(\partial T_x) = \sum_{v \in \partial T_x} \Prob(v)$.
Thus, we have a one-to-one correspondence between probability distributions on
$\partial T$ and forward Markov chains on the rooted tree~$T$. See \cite{RuMa} for an introduction to these ideas.
More generally, we can formulate the following.

\begin{dfn}\label{def:section}
 A \emph{cross section} or \emph{cut} is a subset $S \subset T$ such that for every leaf
$v \in \partial T$, the geodesic $\pi(o,v)$ meets $S$ in precisely one point.
We write $T^S$ for the sub-tree with root $o$ consisting of all vertices that lie on
some geodesic $\pi(o,s)$ with $s \in S$.
\end{dfn}

We have  $\partial T^S = S$ for any section. We define $\Prob_S$ on $S$
in the same way as $\Prob$ was defined on $\partial T$ in \eqref{eq:P}, so that 
$\Prob_S(x) = \Prob(\partial T_x)$. We get the hitting
distribution of the forward Markov chain on $S$, a probability distribution.

Given our length function $\ell$, the \emph{expected length} of a section $S$ with respect
to $\Prob_S$ is 
$$
\ell(\Prob_S) = \sum_{x \in S} |x|_{\ell} \,\Prob_S(x)\,,\quad \text{in particular,}
\quad \ell(\Prob) = \ell(\Prob_{\partial T}) 
$$
In \S \ref{sec:local} -- \S \ref{sec:compare}, 
we shall always assume that $\ell(\Prob) < \infty$.
At last, given the probability distribution $\Prob$ on $\partial T$, we define a new 
\emph{(node average)} measure 
$\mu_{\Prob}$ on $\interior T$ 
by
\begin{equation}\label{eq:nu}
 \mu_{\Prob}(x) = \frac{\ell(x)}{\ell(\Prob)} \, \Prob(\partial T_x) \,.
\end{equation}
It is a probability measure: 
we use the fact that $v \in \partial T_x \iff x \in \pi(o,v^-)$ 
for any leaf $v$ and internal node $x$, whence
$$
\sum_{x \in \interior T} \ell(x)\, \Prob(\partial T_x) 
= \sum_{x \in \interior T} \ell(x) \sum_{v \in \partial T_x} \Prob(v)
= \sum_{v \in \partial T} \Prob(v) \sum_{x \in \pi(o,v^-)} \ell(x) = \ell(\Prob)\,.
$$ 

\medskip
\goodbreak

\noindent\textbf{C. Gradient, Laplacian, and leaf-node interchange}

\smallskip

\begin{dfn}\label{def:lap}
 The \emph{gradient} $\nabla$ and \emph{Laplacian} $\Delta = \Delta_{\Prob}$ 
on $T$ are as follows. For a function $f: T \to \R$ 
$$
\nabla f (y) = \frac{f(y)-f(y^-)}{\ell(y^-)}\,,\; y \in T \setminus \{o\}\,,
\AND 
\Delta f(x) = \sum_{y\in N(x)} \nabla f(y) \, p(y|x)\,.
$$
\end{dfn}

Some remarks are due:  here, our notation differs from \cite{BoeAm} as well as the old note
\cite{RuMa}, who write $\Delta$ for our $\nabla$.
In general, one should see $\nabla f$ as a function on $E(T)$, where $f$ could be interpreted
as a voltage function For $e=[y^-,y]$, the difference $f(y)-f(y^-)$ is the voltage
drop, and thinking of $\ell(e)$ as the resistance of the edge, 
$\nabla f(e) = \bigl(f(y)-f(y^-)\bigr)\big/\ell(e)$ can be interpreted as Ohm's law.
With our assumption, $\ell(e) = \ell(y^-)$, and we can adopt the above notation.
Regarding the Laplacian
$$
\Delta f(x) = \frac{1}{\ell(x)}\sum_{y \in N(x)} \bigl(f(y)-f(x)\bigr) \, p(y|x)\,,
$$
as long as the forward degrees are finite, it is well defined for any function $f$.
In presence of infinitely many successors, we require that the series converges
absolutely:
\begin{equation}\label{eq:absconv}
  |\Delta_{\Prob}|f(x) := \frac{1}{\ell(x)} \sum_{y\in N(x)} \bigl|f(y)-f(x)\bigr| \, p(y|x) < \infty \quad  
\text{for every }\; x \in \interior T\,.
\end{equation}
Our notation is compatible with the one of reversible Markov chain theory,
see e.g. {\sc Woess}~\cite[Chapter 4]{WMarkov}, although the present forward Markov
chains are  not reversible. 

For a probability distribution $\nu$ on any finite or countable set $\mathcal{X}$,
and any function $f : \mathcal{X} \to \R$, we write
$$
\Ex_{\nu}(f) = \sum_{x \in \mathcal{X}} f(x)\, \nu(x)\,,
$$
whenever $f$ is \emph{$\nu$}-integrable, that is $\Ex_{\nu}(|f|) < \infty\,$.
In this setting, the ``Leaf-average node-sum interchange theorem'' (LANSIT) of \cite{RuMa}
is the following. 

\begin{lem}\label{lem:lansit} Let $\Prob$ be a probability distribution on $\partial T$ and
$p(\cdot|\cdot)$ the associated forward Markov chain. 
Let $f: T \to \R$ be a function which satisfies 
$
\Ex_{\mu_{\Prob}}\bigl(|\Delta_{\Prob}|f\bigr) < \infty\,.
$
Then $f$ is $\Prob$-integrable on~$\partial T$, and
$$
\frac{1}{\ell(\Prob)}\, \Ex_{\Prob}\bigl(f - f(o)\bigr) = \Ex_{\mu_{\Prob}}\bigl(\Delta_{\Prob}f\bigr)
$$
\end{lem}

We remark here that in \cite{RuMa}  (where $\ell = \ell_{\sharp}$) 
it is suggested in the last lines of p. 347 that
already the condition $\Ex_{\mu_{\Prob}}\bigl(|\Delta_{\Prob}f|\bigr) < \infty$
is sufficient for the conclusion of their Theorem 1. But this appears to be problematic,
since it does not allow to separate in two parts the sum as displayed in the
middle of the proof on that page. For safety's sake, we present our proof.

\begin{proof}[Proof of Lemma \ref{lem:lansit}]
We start with the right hand side:
$$
\begin{aligned}
\Ex_{\mu_{\Prob}}\bigl(\Delta_{\Prob}f\bigr) 
&\stackrel{(1)}{=} \frac{1}{\ell(\Prob)} \sum_{x \in \interior T} \sum_{y \in N(x)} 
\bigl( f(y) - f(x) \bigr) p(y|x) \Prob(\partial T_x) \\
&\stackrel{(2)}{=} \frac{1}{\ell(\Prob)} \sum_{y \in T \setminus \{o\}} 
\bigl( f(y) - f(y^-) \bigr) \Prob(\partial T_{y}) \\
&\stackrel{(3)}{=} \frac{1}{\ell(\Prob)} \sum_{y \in T \setminus \{o\}} \sum_{v \in \partial T_y}
\bigl( f(y) - f(y^-) \bigr)  \Prob(v)\\
&\stackrel{(4)}{=} \frac{1}{\ell(\Prob)} \sum_{v \in \partial T} \sum_{y \in \pi(o,v) \setminus \{o\}}
\bigl( f(y) - f(y^-) \bigr)  \Prob(v) 
= \frac{1}{\ell(\Prob)} \sum_{v \in \partial T} \bigl( f(v) - f(o) \bigr)  \Prob(v)\,.
\end{aligned}
$$
By assumption, the right hand side of (1) is absolutely convergent, so that the sum
can be rearranged to get (2) and thus (3), which are also absolutely convergent. This
allows to exchange the order of summation,
leading to (4).
\end{proof}


\section{Entropies and tightness}\label{sec:entropies}

\noindent\textbf{A. Entropies and relative entropies}

\smallskip

The \emph{entropy} of our probability measure $\Prob$ on $\partial T$ is
$$
\HH(\Prob)= \sum_{v \in \partial T} \Prob(v) \log_2 \frac{1}{\Prob(v)}\,.
$$
When we have infinite forward degrees in
$T$ then $\partial T$ is infinite. In this case, \emph{we assume that\/} $\HH(\Prob) < \infty$.
At each inner node $x$ of $T$, we have the outgoing 
probability distribution $p(\cdot|x)$ on the set of successors of $x$.
Its entropy is defined analogously,
$$
\HH_p(x) = \HH\bigl(p(\cdot|x)\bigr) = \sum_{y\in N(x)} p(y|x) \log_2 \frac{1}{p(y|x)}.
$$
\begin{lem}\label{lem:entropy} Our assumption that  $\HH(\Prob) < \infty$ 
implies finiteness of $\HH_p(x)$ for each $x$, and
$$
\frac{\HH(\Prob)}{\ell(\Prob)} =
\Ex_{\mu_{\Prob}}\!\!\left( \frac{\HH_p(x)}{\ell(x)}\right).
$$
\end{lem}

On the right hand side, $x$ is the variable with respect to which the expectation is taken.

\begin{proof}  As in \cite{RuMa}, we set $f(x) = -\log_2 \Prob(\partial T_x)$.
Then $\Delta_{\Prob}f \ge 0$, so that the proof of Lemma \ref{lem:lansit}
goes through no matter whether the involved expectations are finite or not.
Thus, the stated formula holds, and when $\HH(\Prob) < \infty$, then also $\HH_p(x)$
is finite for all $x$.
\end{proof}

Note in particular the following special choice of the length function.
\begin{equation}\label{eq:entro-length}
 \text{If} \quad \ell(x)=\HH_p(x) \quad \text{then}\quad \ell(\Prob) = \HH(\Prob)\,.
\end{equation}


Now let us take a second probability measure $\Q$ on $\partial T$, and denote the
transition probabilities of the associated forward Markov chain by $q(y|x)$, when
$y^- = x$. The well-known \emph{Kullback-Leibler divergence} or
\emph{relative entropy} of $\Prob$ with respect to $\Q$ is
$$
\DD(\Prob\|\Q) = \sum_{v \in \partial T} \Prob(v) \log_2 \frac{\Prob(v)}{\Q(v)}\,.
$$
For our purposes, it will always suffice to consider the non-degenerate situation
when $\Prob(v)\,,\; \Q(v) > 0$ for each leaf $v$. In the case when $\partial T$ is
infinite, we assume again that the above series converges absolutely.
Similarly, we have the local version
$$
\DD_{p,q}(x) = \DD\bigl(p(\cdot|x) \| q(\cdot|x)\bigr) \,,\quad x \in \interior T\,.
$$
We need again absolute convergence of its defining series. The following generalises
the corresponding result of \cite{BoeAm} to the case where $T$ has vertices with
infinite forward degree.

\begin{pro}\label{pro:div} 
If $\HH(\Prob) < \infty$ and the defining series of $\DD(\Prob\|\Q)$ is  absolutely convergent, 
then the same holds for each $\DD_{p,q}(x)$, and 
$$
\displaystyle \frac{\DD(\Prob\|\Q)}{\ell(\Prob)} =
\Ex_{\mu_{\Prob}}\!\!\left( \frac{\DD_{p,q}(x)}{\ell(x)}\right).
$$
\end{pro}

Again, $x$ is the variable with respect to which the expectation is taken on the right hand side.

\begin{proof} We want to apply Lemma \ref{lem:lansit} to the function
$f(x) = \log_2 \bigl(\Prob(\partial T_x) / \Q(\partial T_x)\bigr)$.
For this function, we compute
$$
\begin{aligned} 
\ell(\Prob)\, &\Ex_{\mu_{\Prob}}\bigl(|\Delta_{\Prob}|f\bigr) =
\sum_{x \in \interior T} \Prob(\partial T_x) 
\sum_{y \in N(x)} p(y|x) \left| \log_2 \frac{p(y|x)}{q(y|x)} \right|\\
&= \sum_{y \in T \setminus \{o\}} \sum_{v \in \partial T_y} \Prob(v) 
\left| \log_2 \frac{p(y|x)}{q(y|x)} \right|\\
&\le \sum_{v \in \partial T} \Prob(v) \sum_{y \in \pi(o,v) \setminus \{o\}}
\left(\log_2 \frac{1}{p(y|y^-)} + \log_2 \frac{1}{q(y|y^-)}\right)\\
&=\sum_{v \in \partial T} \Prob(v) \left(\log_2 \frac{1}{\Prob(v)} + \log_2 \frac{1}{\Q(v)}\right)\\
&\le \sum_{v \in \partial T} \Prob(v) \left(2 \, \log_2 \frac{1}{\Prob(v)} 
+ \left|\log_2 \frac{\Prob(v)}{\Q(v)}\right|\right) 
= 2\HH(\Prob) + \sum_{v \in \partial T} \Prob(v)  \left|\log_2 \frac{\Prob(v)}{\Q(v)}\right|, 
\end{aligned}
$$
which is finite by assumption. In particular, we get for each internal node $x$
that $|\Delta_{\Prob}|f(x) < \infty$, that is, $\DD_{p,q}(x)$ converges absolutely.  
\end{proof}

In the following two lemmas we state some 
elementary properties of the variational distance between two probability distributions, and
the defining function of the entropy.

\begin{lem}\label{lem:norm} For two probability distributions $\nu_1\,$, $\nu_2$ on any finite or 
countable set $\mathcal{X}$, let
$$
\| \nu_1 - \nu_2 \|_1 = \sum_{x \in \mathcal{X}} |\nu_1(x) - \nu_2(x)|
$$
Then we have the following.
\\[8pt]
\emph{(a)} \hspace*{1cm}$\displaystyle
\sum_{x\,:\, \nu_1(x) > \nu_2(x)} \!\!\bigl(\nu_1(x) - \nu_2(x)\bigr) = 
\sum_{x\,:\, \nu_1(x) < \nu_2(x)} \!\!\bigl(\nu_2(x) - \nu_1(x)\bigr) = 
\frac{1}{2} \| \nu_1 - \nu_2 \|_1\,.
$
\\[8pt]
\emph{(b) [Pinsker's inequality]} \hspace*{1cm}
$\displaystyle \| \nu_1 - \nu_2 \|_1^2 \le 2 \ln 2 \,\, \DD(\nu_1 \| \nu_2)\,.$
\end{lem}

The optimal constant $2\ln 2$ is due to Csisz\'ar, Kemperman, and Kullback in independent
work of the 1960ies. There are also converse inequalities; see 
{\sc Csisz\'ar and Talata}~\cite{CsTa} and
{\sc Sason and Verd\'u}~\cite{SaVe} and the references given there. The following is elementary.

\begin{lem}\label{lem:phi}
The function $\varphi(t) = t \log_2(1/t)$, $t\in [0\,,\,1]$, with the convention 
$\varphi(0)=0$, has the following properties.
\begin{enumerate}
 \item[(i)] $\varphi$ is strictly concave and assumes its maximal value $1/(e\ln 2)$ 
at $t=1/e$.\\[-1pt]
\item[(ii)] For $0 \le t < t+\delta\le 1$, we have 
$|\varphi(t+\delta)-\varphi(t)| \le \max \{ \varphi(\delta), \varphi(1-\delta) \}$.\\[-1pt]
\item[(iii)] For $0 \le \delta\le 1/2$, we have 
$\max \{ \varphi(\delta), \varphi(1-\delta) \} = \varphi(\delta)$.\\[-1pt]
\item[(iv)] For $t, u \in [0\,,\,1/(2e)]$, we have 
$\varphi(t+u) \le 2\bigl(\varphi(t)+\varphi(u)\bigr)$.
\end{enumerate}
 \end{lem}


 \medskip

\noindent\textbf{B. Tightness and uniform summability}

Given our tree, suppose that for each internal node $x$, we have a real valued function 
$y \mapsto g(y|x)$, defined on the successors $y$ of $x$. We are interested
in the series $\sum_{y\in N(x)} g(y|x)$. Then we say that the family of functions,
resp., the associated series are \emph{tight} or \emph{uniformly (absolutely) summable,}
if for every $\ep > 0$ there is an integer $M_{\ep} \ge 1$ such that for
each $x \in \interior T$, there is a subset $N_{\ep}(x) \subset N(x)$
such that
\begin{equation}\label{eq:tight}
 |N_{\ep}(x)| \le M_{\ep} \AND \sum_{y \in N(x) \setminus N_{\ep}(x)} \bigl|g(y|x)\bigr| < \ep\,.
\end{equation}
(The terminology ``tight'' is commonly used in presence of a family of probability
distributions.) As long as the tree is finite,
\eqref{eq:tight} is always valid. In this case, we need no $\ep$, i.e., we can replace
$N_{\ep}(x)$  by all of $N(x)$ and $M_{\ep}$ by 
\begin{equation}\label{eq:M}
M = \max \{ \dg(x): x \in \interior T\}.
\end{equation}
When there are only finitely many types of different functions
$g(\cdot|x)$ up to bijections between the sets $N(x)$, $x \in \interior T$, 
then \eqref{eq:tight} is equivalent with absolute convergence of the involved
series.

In our context, we shall want that all the sums
\begin{equation}\label{eq:sums}
\HH_p(x) = \sum_y p(y|x) \log \frac{1}{p(y|x)}\AND
\HH_q(x) = \sum_y q(y|x) \log \frac{1}{q(y|x)} 
\end{equation}
are uniformly summable.  Uniform summability implies with a short computation that 
$\HH_q(x) \le \log_2 M_{\ep} + 1/ (e \ln 2) +\ep$. In particular, 
if $\inf \{\ell(x) : x \in \interior T\} > 0$ 
then $A = \sup \{\HH_q(x)/\ell(x) : x \in \interior T\} < \infty\,$. This will be needed further
below.

We 
want to formulate sufficient conditions for uniform summability which are easy to grasp. 
We can enumerate each set $N(x) = \{ y_n(x) : n \in \N\,, n \le \dg(x) \}$ (finite or countable) and set
$$
p_x(n) = p\bigl(y_n(x)|x) \;\text{ for }\; n \le \dg(x) 
\AND p_x(n) = 0 \;\text{ for }\; n > \dg(x)
$$
(the latter when $\dg(x)$ is finite). Then tightness of the family of distributions
$p_x(\cdot)$, $x \in \interior T$, means that we can do the above enumerations in such 
a way that
\begin{equation}\label{eq:tails}
\tau(k) = \sup \{\tau_x(k) : x \in \interior T\} \to 0 \;\text{ as }\; k \to \infty\,,
\quad\text{where}\quad \tau_x(k) = \sum_{n=k}^{\infty} p_x(n)\,.
\end{equation}
In this case, $p(n) = \tau(n) - \tau(n+1)$ defines a probability distribution on $\N$,
which \emph{stochastically dominates} each distribution $p_x$, that is, its \emph{tails}
$\tau(k)$ majorise the tails of all $p_x\,$.

\begin{lem}\label{lem:moment} Suppose that the distributions $p_x = p(\cdot| x)$ 
satisfy \eqref{eq:tails}
and that the dominating distribution $p$ has finite mean,
$$
\sum_{n=1}^{\infty} n\,p(n) < \infty\,.
$$
Then the entropies $\HH_p(x)$, $x \in \interior T$, are uniformly summable.
 
If in addition the two probability measures $\Prob$ and $\Q$ on $\partial T$ are such that
\begin{equation}\label{eq:equiv}
\frac{1}{c} \cdot \Q \le \Prob \le c \cdot \Q
\end{equation}
for some finite $c > 0$, then also the entropies $\HH_x(q)$, $x \in \interior T$, are uniformly summable.
%
\end{lem}

\begin{proof} By Lemma \ref{lem:phi}, the function $\varphi(t)$ is monotone increasing
in the interval $[0\,,\,1/e]$. 
If $n \ge 2$ and $p_x(n) \le 2^{-n}$ then $\varphi\bigl(p_x(n)\bigr) \le n 2^{-n}$.
Otherwise, $\varphi\bigl(p_x(n)\bigr) \le n p_x(n)$. Choose $m \ge 2$.  We have
$$
\sum_{n=m}^{\infty} np_x(n) = m \tau_x(m) + \sum_{k=m+1}^{\infty} \tau_x(k)
\le m \tau(m) + \sum_{k=m+1}^{\infty} \tau(k) = \sum_{n=m}^{\infty} np(n)\,.
$$
Consequently, 
$$
\sum_{n=m}^{\infty} \varphi\bigl(p_x(n)\bigr) \le \sum_{n = m}^{\infty} n \bigl(2^{-n} + p(n)\bigr)\,,
$$
which tends to $0$ as $m \to \infty$. This proves uniform summability. The second part is obvious.
\end{proof}

\section{A comparison theorem -- discussion and variations}\label{sec:compare}

Here is our first main result.

\begin{thm}\label{thm:compare}
Let $\Prob$ and $\Q$ be probability measures supported by all leaves in $\partial T$,
and let $p(\cdot|\cdot)$ and $q(\cdot|\cdot)$ be the transition kernels of the respective 
associated forward Markov chains. Suppose that $\HH(\Prob)$, $\HH(\Q)$ and $\DD(\Prob\|\Q)$ are
finite and that the defining series of $\HH_p(x)$ and $\HH_q(x)$, 
$x \in \interior T$,
are uniformly summable (without loss of generality with the same value $M_{\ep}$ and sets 
$N_{\ep}(x)$ for every
$\ep > 0$). 

Then, for any choice of $\ep > 0$ and $0 < \delta < 1/2$,
$$
\left| \frac{\HH(\Prob)}{\ell(\Prob)} - \frac{\HH(\Q)}{\ell(\Q)}\right| \le
L \Bigl( 2\ep + M_{\ep}\,\varphi(\delta) \Bigr)
+ \frac{C\sqrt{L}}{\delta}\sqrt{\frac{\DD(\Prob\|\Q)}{\ell(\Prob)}}
+ \frac{A-a}{2} \, \|\mu_{\Prob} - \mu_{\Q}\|_1\,,
$$
where
\begin{align*}
L &= \frac{\ell_{\sharp}(\Prob)}{\ell(\Prob)}\,,\qquad C = \frac{2\sqrt{2}}{e\sqrt{\ln 2}} \approx 1.25\,,\\
A &= \sup\left\{ \frac{\HH_q(x)}{\ell(x)} : x \in \interior T\right\},\AND 
a = \inf\left\{ \frac{\HH_q(x)}{\ell(x)} : x \in \interior T\right\}.
\end{align*}
%
When $T$ has bounded forward degrees, we can set $\ep =0$ and replace the $M_{\ep}$ 
of \eqref{eq:tight} by the $M$ of \eqref{eq:M}.
\end{thm}

The proof of the theorem is deferred to the Appendix.
Right now, we embark on a detailed discussion of various aspects and special cases 
of this inequality. 
As outlined in the Introduction, our way of thinking is that $\Q$ is the basic 
reference measure, while $\Prob$
may vary; it should be considered as a perturbation of $\Q$.

\medskip



We say that two internal nodes $x$ and $x'$ have the \emph{same $\Q$-type,} if
there is a bijection $\psi: N(x) \to N(x')$ such that 
\begin{equation}\label{eq:type}
q(\psi y| x') = q(y|x) \quad \text{for all }\; y \in N(x)\,.
\end{equation}

\medskip

\noindent\textbf{A. Identical $\Q$-types}

\smallskip

The simplest, but most significant situation is the one where all internal nodes 
have the same $\Q$-type.
In this case, $\HH_q(x) = \HH_q$ is constant, and uniform summability of $\HH_q(x)$
reduces to finiteness of $\HH_q\,$. If in addition $\ell = \ell_{\sharp}$,
then $A=a$, and by Lemma \ref{lem:entropy}, one has $\HH(\Q) = \ell_{\sharp}(\Q)\cdot \HH_q\,$.
The inequality of Theorem \ref{thm:compare} becomes
\begin{equation}\label{eq:compare}
 \left| \frac{\HH(\Prob)}{\ell_{\sharp}(\Prob)} - \HH_q\right| \le
2\ep + M_{\ep}\,\varphi(\delta)
+ \frac{C}{\delta}\sqrt{\frac{\DD(\Prob\|\Q)}{\ell_{\sharp}(\Prob)}}\,.
\end{equation}
When $T$ is finite, so that we can set $\ep =0$ and replace $M_{\ep}$ by $M$, 
this is the main result of 
\cite{BoeAm}. \eqref{eq:compare} provides an extension to the case when all $N(x)$
are infinite and there is one $\Q$-type with $\HH_q < \infty$. Note that when in addition
the measures $\Prob$ and $\Q$ are comparable as in \eqref{eq:equiv} then the uniform
summability conditions of Theorem \ref{thm:compare} hold.
  
Suppose furthermore that $T$ has height $n$, that is, $|v|_{\sharp}=n$ for each leaf of $T$.
Then with suitable labelling, $\Q = q^{\otimes n} = q \otimes \cdots \otimes q$ is the $n$-fold product measure
of $q(x) := q(x|o)$, $x \in N(o)$, or in other words, it is the joint distribution of
$n$ independent random variables $X_1, \dots, X_n$ with common distribution $q(\cdot)$, and 
\eqref{eq:compare} becomes
\begin{equation}\label{eq:compare-n}
 \left| \frac{\HH(\Prob)}{n} - \HH_q\right| \le
2\ep + M_{\ep}\,\varphi(\delta)
+ \frac{C}{\delta}\sqrt{\frac{\DD(\Prob\|q^{\otimes n})}{n}}\,.
\end{equation}
Assume that for fixed $n$, we have a sequence $\Prob^{(k)}$, $k=1,2,,\dots$, each of which 
is a perturbation of the joint distribution $\Q$ of $(X_1, \dots, X_n)$ such that
$\DD(\Prob^{(k)}\|q^{\otimes n}) \to 0$ as $k \to \infty$. Then, under the stated conditions 
of uniform summabilty -- in particular,
when $\frac{1}{c} \cdot q^{\otimes n} \le \Prob^{(k)} \le c\cdot q^{\otimes n}$ -- 
$$
\lim_{k \to \infty} \frac{\HH(\Prob^{(k)})}{n} = \HH_q\,.
$$
We shall come back to this further below, considering the situation where $n \to \infty$.

\vfill\eject

\noindent\textbf{B. Variable $\Q$-types}

\smallskip

If $\HH_q(x)$ is not constant, then we may use the length function $\ell_q(x) = \HH_q(x)$ instead
of $\ell_{\sharp}\,$.
Again, $A = a$, and \eqref{eq:entro-length} -- applied to $\Q$ instead of $\Prob$ -- leads to
\begin{equation}\label{eq:compare-q}
 \left| \frac{\HH(\Prob)}{\ell_{q}(\Prob)} - 1\right| \le
2\ep + M_{\ep}\,\varphi(\delta)
+ \frac{C}{\delta}\sqrt{\frac{\DD(\Prob\|\Q)}{\ell_{q}(\Prob)}}\,.
\end{equation}
If $\HH_q(x)$ varies between two positive bounds, then $\ell_{q}(\Prob)/\ell_{\sharp}(\Prob)$
varies between the same bounds, i.e., the two expected lengths are of the same order of
magnitude.

\smallskip

Otherwise, e.g. if we want to stick to $\ell_{\sharp}\,$, we need to discuss the 
quantity $\frac{A-a}{2}\,\|\mu_{\Prob} - \mu_{\Q}\|_1 \,$, since
it remains present in the estimate of Theorem \ref{thm:compare}.
%
We can use the following variant.

\begin{lem}\label{lem:variant} Let $S$ be a cross section of $T$ (Def. \ref{def:section}),
and let
$$
A_S^* = \sup\left\{ \frac{\HH_q(x)}{\ell(x)} : x \in \interior T \setminus \interior T^S\right\},
\AND 
a_S^* = \inf\left\{ \frac{\HH_q(x)}{\ell(x)} : x \in \interior T \setminus \interior T^S\right\}.
$$
Then the term $\frac{A-a}{2} \, \|\mu_{\Prob} - \mu_{\Q}\|_1$
in the estimate of Theorem \ref{thm:compare} can be replaced with
$$
\frac{A_S^*-a_S^*}{2} \, \|\mu_{\Prob} - \mu_{\Q}\|_1 + 
(A-a) \,\max \left\{ \frac{\ell(\Prob_S)}{\ell(\Prob)}\,,\, 
                     \frac{\ell(\Q_S)}{\ell(\Q)} \right\}\,.
$$
\end{lem}

This is useful, if $H_q(x)/\ell(x)$ is ``almost constant'' outside a small
sub-tree. By this we mean that there is a cross section $S$ such that $A_S^*-a_S^*$ is small,
and at the same time, $T$ is much larger that $T^S$, so that also 
$\ell(\Prob_S)/\ell(\Prob)$ and $\ell(\Q_S)/\ell(\Q)$ are small.
See below in \S \ref{sec:infinite}. The proof of the lemma is also deferred
to the Appendix.

We conclude this section with an example which shows that the term 
$\frac{A-a}{2}\,\|\mu_{\Prob} - \mu_{\Q}\|_1$ is indispensable
in the general estimate of Theorem \ref{thm:compare}.

\begin{exa}\label{ex:indisp}
We consider a tree $T=T^{(n)}$ of height $n+1$ whose root has $2$ forward neighbours,
so that the two branches $T_1$ and $T_2$ of height $n$ which are attached to $o$
have all inner nodes with forward degrees $d_1$ in $T_1$ and $d_2$ in $T_2\,$, respectively.
Then $\partial T = \partial T_1 \cup \partial T_2\,$.
For $0 < \theta < 1$, we define
$$
\Prob_{\theta}(v) = \Prob_{n, \theta}(v) =
                     \begin{cases} \theta/d_1^{\,n}\,,&\text{if}\; v \in \partial T_1\,,\\
                                   (1-\theta)/d_2^{\,n}\,,&\text{if}\; v \in \partial T_2\,.
                     \end{cases}
$$
With $\ell=\ell_{\sharp}$, we have $\ell(\Prob_{\theta}) = n+1$.
It is easy to compute
$$
\HH(\Prob_{\theta}) = \HH(\theta,1-\theta) + n
\Bigl(\theta \,\log_2 d_1 + (1-\theta)\,\log_2 d_2\Bigr)\,,
$$
where $\HH(\theta,1-\theta)$ is the entropy of the Bernoulli distribution with
parameter $\theta$.

Now we take $\Prob_n = \Prob_{n,\theta}\,$, where $\theta \ne 1/2$, and $\Q_n = \Prob_{n, 1/2}\,$.
Then it is also easy to compute 
$$
\DD(\Prob_n\|\Q_n) = 1 - H(\theta,1-\theta)
$$
As $n \to \infty$, we have that 
$$
\frac{\DD(\Prob_n\|\Q_n)}{n} \to 0\,, 
\quad \text{while}\quad \left| \frac{\HH(\Prob_n)}{n} - \frac{\HH(\Q_n)}{n} \right| 
\to \bigl|\bigl(\theta -  \tfrac{1}{2}\bigr)\log_2 d_1 
    + \bigl(1- \theta -  \tfrac{1}{2}\bigr)\log_2 d_2\bigr|\,,
$$
and the latter limit is non-zero for appropriate choices of $\theta$, $d_1$ and $d_2\,$.
Thus, there can be no improved general estimate that can remove the term 
$\frac{A-a}{2}\,\|\mu_{\Prob} - \mu_{\Q}\|_1 \,$ or replace it by something of significantly
smaller order of magnitude.\qed
\end{exa}

The last example can be seen in terms of two stationary stochastic processes on denumerable
state spaces. We first toss a coin to decide which of the two processes to run, and
then ``forget'' about the other one. The reference measure $\Q$ corresponds to a fair
coin toss, while $\Prob$ corresponds to a biased one.  Then we cannot compare their 
entropy rates just in terms of their Kullback-Leibler divergence, as 
we see from the above example. 
 
\section{Trees with infinite geodesics}\label{sec:infinite}

Now, and only in this last chapter before the Annex, we consider the situation where 
our tree does have infinite geodesics.
In their presence, let us describe the boundary $\partial T$ of $T$.
It consists of eventual leaves (vertices $v \ne o$ with no forward
neighbour) plus the \emph{ends} of $T$. Each end $w$ is represented by an
infinite geodesic $\pi(o,w)$ starting at $o$. 
We can put a metric on 
$\wh T = \interior T \cup \partial T$: for any pair of distinct points 
$v, w \in \wh T$, their \emph{confluent} $v \wedge w$ is the last common
vertex on the geodesics $\pi(o,v)$ and $\pi(o,w)$.
Then the metric is
$$
d(v,w) = \begin{cases} 0\,,\;&\text{if}\; v=w\,,\\
                       2^{-|v\wedge w|} \,,\;&\text{if}\; v\ne w\,. 
         \end{cases}
$$
Thus, a sequence $(x_n)$ of vertices of $T$ converges to an end $w$,
if $|x_n \wedge w| \to \infty\,$. For more details in a context 
close to the present note, see e.g. \cite[Ch.9.B]{WMarkov}.
Note that in many typical examples, such as the infinite binary tree,
$\partial T$ is a Cantor set, whence uncountable.

For $x \in \interior T$, the branch $T_x$ and its boundary $\partial T_x$ are 
defined as before. The boundary now may contain ends as well as leaves, and the sets
of all  $\partial T_x\,$, $x \in \interior T$, are the open (and also closed!)
balls for the metric on $\partial T$. Any probability measure
$\Prob$ on $\partial T$ is given via consistent definition
of the values $\Prob(\partial T_x)\,$, $x \in \interior T$. That is, we 
need to have
$$
\Prob(\partial T)=1 \AND \sum_{y \in N(x)} \Prob(\partial T_y) = \Prob(\partial T_x)
\;\text{ for every } x \in \interior T\,. 
$$

\begin{exa}\label{ex:proc}
The most typical class of examples arises from a stochastic process $(X_n)_{n \ge 1}$
on a finite or countable state space $\mathfrak{S}$. In this case, for $n \in \N$,
let 
$$
p_n(s_1, \dots, s_n) = \Prob[X_1 = s_1\,,\dots, X_n=s_n]\,,\quad s_1, \dots, s_n \in \mathfrak{S}
$$
be the joint distribution of $(X_1, \dots, X_n)$. Then we let (the vertex set of) our tree
be 
\begin{equation}\label{eq:traj}
T = \{ s_1s_2\cdots s_n : n \ge 0\,,\; s_i \in \mathfrak{S}\,,\; p_n(s_1, \dots, s_n) > 0 \}\,.
\end{equation}
For $n=0$, this includes the \emph{empty word} $o$, for which $p_0(o)=1$.
The predecessor vertex of $x=s_1\cdots s_n$ ($n\ge 1$) is $x^- = s_1 \cdots s_{n-1}$,
and
$$
p(x| x^-) = \Prob[X_1 = s_1\,,\dots, X_{n-1}=s_{n-1}\,,X_n=s_n \mid X_1 = s_1\,,\dots, X_{n-1}=s_{n-1}]\,,
$$
while for $x = s \in N(o)$, we have $p(x|o) = \Prob[ X_1=x]$. 
In this case, $T$ has no leaves, the end space 
$$
\partial T = \{ s_1s_2s_3\cdots\, : \,p_n(s_1, \dots, s_n) > 0 \; \text{for each}\; n \}
$$
is the space of active trajectories of our process, and $\Prob$ is the distribution
of that process.\qed
\end{exa}

In general, we cannot speak directly about the entropy of $\Prob$ on $\partial T$. We can
take any section $S$ of $T$ and consider $\HH(\Prob_S)$, as in Definition \ref{def:section}
and the subsequent lines. We restrict attention to the
following sequence of sections with respect to the standard length function:
$$
S(n) = \{ x \in T : |x| = n \} \cup \{ v \in \partial T : |v| < n \}.
$$ 
The second part appears when $T$ has leaves.  We write 
$$
\Prob_n = \Prob_{S(n)} \quad\text{and}\quad \Q_n = \Q_{S(n)}
$$ and assume that $\HH(\Prob_n)$ and $\DD(\Prob_n\|\Q_n)$ are finite (absolutely
convergent) for each $n$. Then we can speak about the quantity
$$
h= h\bigl(\Prob,\ell) = \lim_{n \to \infty} \frac{\HH(\Prob_n)}{\ell(\Prob_n)}\,,
$$
whenever the limit exists. This is a variant of the entropy rate.
In particular, in the situation of Example \ref{ex:proc}, for the standard length function,
$h(\Prob) = h(\Prob,\ell_{\sharp})$ is indeed the classical  entropy rate of 
the process $(X_n)$ with distribution $\Prob$. 
%
%
In any case, relative entropy is well defined; see {\sc Gray}~\cite[\S 7.1]{Gr}:
\begin{equation}\label{eq:DPQ-gen}
\DD(\Prob\|\Q) = \lim_{n \to \infty}  \DD(\Prob_n\|\Q_n)
\end{equation}
Indeed, one can apply the well-known log-sum inequality --
see {\sc Cover and Thomas}~\cite[Thm. 2.7.1]{CoTh} --
to verify that  $\DD(\Prob_n\|\Q_n)$ is increasing with $n$. Of course,
the limit may be infinite. For the proof of the next theorem, see once more the Appendix.

\begin{thm}\label{thm:asy}
Let $T$ be a countable tree with infinite geodesics and $\Prob$ and $\Q$
two probability measures giving positive mass to the set of ends of $T$ and satisfying
our tightness requirements for the local entropies at the inner vertices. 
Assume that
\begin{itemize}
 \item[(i)] The length function $\ell$ is comparable with $\ell_{\sharp}\,$, that is
$0 < c_1 \le \ell(x) \le c_2 < \infty$ for suitable $c_1, c_2$ and all $x \in \interior T$,  
\item[(ii)] the limit $\;\displaystyle h = \lim_{|x| \to \infty} \HH_q(x)/\ell(x)\;$
exists and is finite, and 
\item[(iii)] for the relative entropies,  $\; \displaystyle \lim_{n \to \infty}
\DD(\Prob_n\|\Q_n)/n =0$. 
\end{itemize}
Then
$$
h(\Prob,\ell) = h(\Q,\ell) = h\,.
$$
\end{thm}

Note that condition (iii)  holds in particular
when $\DD(\Prob\|\Q) < \infty\,$.

\begin{app}\label{app:proc-cont} We take up Example \ref{ex:proc}, with the difference
that we now write $\Q$ for the distribution on the trajectory space of
our random process $(X_n)$. We have identified that space with the space of ends
of a rooted tree $T$ which has no leaves. 
We take the natural length function $\ell_{\sharp}$ and assume that $H_q(x) = h$ is constant.
Then the standard entropy rate $h(\Q)$ exists and coincides with $h$.

Now we consider a perturbation $(\wt X_n)$ of $(X_n)$ with overall distribution $\Prob$.
In this situation, Theorem \ref{thm:asy} yields that if $\DD(\Prob_n\|\Q_n)/n \to 0$,
and in particular if $\DD(\Prob\|\Q) < \infty$,  then 
also the entropy rate of $(\wt X_n)$ exists and is equal to the same $h$.

Typical cases where this can be used are as follows.

\smallskip\noindent
\textbf{(a)}$\;$
As in \eqref{eq:compare-n}, let $(X_n)$ be a sequence of i.i.d. discrete
random variables with common distribution $q(\cdot)$. Then $\Q = q^{\otimes \N}$ 
is the infinite product measure of countably many copies of $q$, and the
entropy rate is $h(\Q) = \HH(q)$. If we have a perturbed process $(\wt X_n)$
whose overall distribution $\Prob$ satisfies  $\DD(\Prob_n\|q^{\otimes n})/n \to 0$
then also $h(\Prob) = \HH(q)$. (When $q(\cdot)$ is finitely supported, this is
covered by \cite{BoeAm}.)

The following specific example is related to the famous paper of {\sc Kakutani}~\cite{Ka}
on absolute continuity versus orthogonality of infinite product measures. We let $q(\cdot)$
be uniform distribution on $\{ 1, \dots, M \}$ and $\Q = q^{\otimes \N}$. The associated
tree is the $M$-ary tree, where each vertex has $M$ successors. For $\alpha \in (0\,,\,1)$,
we define $p_{\alpha}(\cdot)$ on $\{ 1, \dots, M \}$ by 
$$
p_{\alpha}(1) = \frac{1+\alpha}{M}\,,\; p_{\alpha}(2) = \frac{1-\alpha}{M}\,,\quad
\text{and}\quad p_{\alpha}(s) = \frac{1}{M} \; \text{ for } s = 3,\dots, M\,.
$$
Then
$$
\DD(p_{\alpha}\|q) = \HH(q) - \HH(p_{\alpha}) = \frac{f(\alpha)}{M}\,,\quad\text{where}\quad
f(\alpha) = (1+\alpha)\,\log_2(1+\alpha) + (1-\alpha)\,\log_2(1-\alpha).
$$
We note that $f(\alpha)/\alpha^2 \to \ln 2$ as $\alpha \to 0$.
Now we take a sequence $(\alpha_n)$ and the perturbed measure 
$
\Prob = \bigotimes\limits_{n=1}^{\infty} p_{\alpha_n}\,.
$
Then 
$$
\DD(\Prob_n\|\Q_n) = \HH(\Q_n) - \HH(\Prob_n) 
= \frac{f(\alpha_1) + \dots + f(\alpha_n)}{M} 
$$
If $\alpha_n \to 0$ then $\DD(\Prob_n\|\Q_n)/n \to 0$, whence
$h(\Prob) = h(\Q) = \log_2 M$. If in particular $\alpha_n = n^{-\beta}$ with $\beta > 0$
then we see that $\DD(\Prob\|\Q) < \infty$ precisely when $\beta \ge 1/2$.  

\smallskip\noindent
\textbf{(b)}$\;$
Next, let $(X_n)$ be a Markov chain on a finite or countable
state space $\mathfrak{S}$ with transition matrix 
$\mathfrak{Q} = \bigl(\mathfrak{q}(s'|s)\bigr)_{s,s' \in \mathfrak{S}}$
and initial distribution $\nu$ on $\mathfrak{S}$.
Denote once more by $\Q$ -- instead of $\Prob$ as in Example \ref{ex:proc} --
the overall probability distribution of the process on its trajectory space,
which we interpret as the boundary at infinity of our tree $T$. In this case,
as in \eqref{ex:proc} but with $q_n$ instead of $p_n$, we have
$$
q_n(s_1\,,\dots, s_n) = \mathfrak{q}(s_n|s_{n-1})\cdots\mathfrak{q}(s_2|s_1) \nu(s_1)\,.
$$
In particular, for a node $x=s_1 \cdots s_{n-1}s_n$ of our tree, with $n \ge 2$,
we have $q(x|x^-) = \mathfrak{q}(s_n|s_{n-1}).$

Now suppose that the transition matrix is such that all its rows $\mathsf{q}(\cdot|s)$
have the same finite entropy $h$. Then $h(\Q)=h$. Again, if we have a perturbation
$(\wt X_n)$ of $(X_n)$ with overall distribution $\Prob$ such that our tightness
requirements are met and $\DD(\Prob_n\|\Q_n)/n \to 0$, then
also $h(\Prob)=h$.

\smallskip\noindent
\textbf{(c)}$\;$A specific example concerning the last situation is that of a
finite or infinite oriented graph with vertex set $\mathfrak{S}$ which is $\mathsf{d}$-regular,
that is, every vertex has $\mathsf{d}$ outgoing edges. For our Markov chain, we
take \emph{simple random walk} which at each vertex moves with equal probability
to one of those neighbours. Then $h = \log_2 \mathsf{d}$ and all of the above
applies. 

There are many further examples of Markov chains where $\HH\bigl(\mathsf{q}(\cdot|s)\bigr)$
is constant, in particular those where all rows $\mathsf{q}(\cdot|s)$
of the transition matrix are permutations of each other.

\smallskip

If more generally $H_q(x)$ is not constant, then we can use $\ell(x) = H_q(x)$ as in
\eqref{eq:compare-q}, and when $\DD(\Prob\|\Q) < \infty$, or just $\DD(\Prob_n\|\Q_n)/n \to 0$
we get that 
$$
a \le \liminf_{n \to \infty} \frac{1}{n}H(\wt X_1\,,\dots, \wt X_n) 
\le \limsup_{n \to \infty} \frac{1}{n}H(\wt X_1\,,\dots, \wt X_n) \le A\,,
$$
where $a$ and $A$ are as in Theorem \ref{thm:compare}.
In this situation, it may also be of interest to note that if 
$h(\Prob) = \lim_{n \to \infty} \frac{1}{n}H(\wt X_1\,,\dots, \wt X_n)$ exists,
then 
$$
h(\Prob) = \lim_{n\to \infty}  \frac{\ell_q(\Prob_n)}{n}\,,
$$
which follows once more from \eqref{eq:compare-q}.
\qed
\end{app}

\begin{rmk}\label{rmk:fin}
All results stated so far adapt easily to the situation where
the reference measure $\Q$ is supported by the whole of $\partial T$,
while the probability measure $\Prob$ is supported by a possibly
strict subset. In this situation, we may have $p(x|x^-)=0$ for some
edges of $T$, and when $\partial T$ consists only of leaves as in \S \ref{sec:local} --
\S \ref{sec:compare}, the node-average measure $\mu_{\Prob}$ will only live
on the induced sub-tree of $T$ which is spanned by all vertices that
can be reached from the root. In the presence of infinite geodesics, the situation is
analogous, considering the measure $\mu_{\Prob_n}$ on $T^{S(n)}$.
\end{rmk}

The following is our final application.

\smallskip\noindent
\textbf{Random perturbations.} We start once more
with the setting of Example \ref{ex:proc}, that is, the
infinite tree is as in \eqref{eq:traj}, and its boundary is equipped with our reference 
probability measure $\Q$. At the internal vertices we have the transition probabilities 
$q(\cdot|x)$ and a length function $\ell(x)$. 

$\Q$ describes the basic stochastic process which then is randomly perturbed: 
there is a second collection of transition probabilities $q'(\cdot|x)$ at each internal
vertex, supported by some or all forward neighbours of $x$. 
We assume that both families of probability distributions $q(\cdot|x)$ and $q'(\cdot|x)$, 
$x \in T$, are tight and that their entropies are uniformly summable;
see Lemma \ref{lem:moment} for a sufficient condition.

Now we consider a sequence of \emph{random variables} $\dde_n\,$, $n \ge 0$,
taking their values in the interval $[0\,,\,1]$, defined on a separate probability space.
Our perturbed random process is given 
via the forward transition probabilities
\begin{equation}\label{eq:random}
p(y|x) = (1-\dde_n)\,q(y|x) + \dde_n\, q'(y|x)\,, \quad\text{if }\; y \in S(n)\,.
\end{equation}
Note that this defines \emph{random} transition probabilities, so that we get a 
randomised process: at time $n$ -- when $x \in S(n)$ -- the moving particle 
chooses to proceed according to the ``wrong'' forward transition probabilities $q'(\cdot|x)$ 
with probability $\dde_n\,$, while it follows the ``correct'' transition rule with the
complementary probability. We write $\Prob$ for the \emph{random} probability measure
on $\partial T$ induced by $p(\cdot|\cdot)$, and as before $\Prob_n = \Prob_{S(n)}\,$.
The  proof of the following is again contained in the Appendix.

\begin{thm}\label{thm:randper}
In the setting of \eqref{eq:random} with tight families of forward transition probabilities
$q(\cdot|\cdot)$ and $q'(\cdot|\cdot)$ as well as uniformly summable local entropies,
assume that
\begin{itemize}
 \item[(i)] The length function $\ell$ is comparable with $\ell_{\sharp}\,$, \\[-1pt]
\item[(ii)] the limit $\;\displaystyle h = \lim_{|x| \to \infty} \HH_q(x)/\ell(x)\;$
exists and is finite, and \\[-1pt]
\item[(iii)] there is $D < \infty$ such that $\DD_{q',q}(x) < D$ for all $x \in T$.
\end{itemize}
Under these conditions, if $\displaystyle \; \lim_{n \to \infty} \Ex(\dde_n) = 0 
\; \text{ then }\; 
\lim_{n \to \infty} \frac{\HH(\Prob_n)}{\ell(\Prob_n)} = h \; \text{ in probability.}
$\\[3pt]
Moreover, if $\;\dde_n \to 0\;$ almost surely, then 
$\;h(\Prob,\ell) = h\;$ almost surely, i.e., for almost every realisation of the sequence
$(\dde_n)$. 
\end{thm}

\begin{exa}\label{ex:last} Suppose that $\dde_n$ is a Bernoulli random variable
 with $\Prob[\dde_n=1] = 1/n^{\beta}$ and $\Prob[\dde_n=0]= 1 - 1/n^{\beta}\,$,
where $\beta > 0$. Then $\Ex(\dde_n) \to 0$, so that we get convergence 
in probability in Theorem \ref{thm:randper}. If $\beta > 1$ then by the 
Borel-Cantelli Lemma we have with probability $1$ that $\dde_n=0$ for all
but a (random) finite number of $n$. In this case, we get
that $h(\Prob,\ell)=h(\Q,\ell)$ almost surely.
\end{exa}

The general picture that we have in mind regarding random perturbations with $\dde_n \to 0$
is that $\Q$ relates to a given process which one wants to simulate. The simulation is then 
a process with overall distribution $\Prob$, and at each step, the error in the simulation
is described via $q'(\cdot|\cdot)$ and $\dde_n$ as in \eqref{eq:random}. Then $\dde_n \to 0$
means that in the course of time, there is a learning effect, so that the simulation
gets better and better.

\section{Appendix}\label{sec:appendix}

In this appendix, we present the proof details of Theorem \ref{thm:compare},
Lemma \ref{lem:variant}, as well as theorems \ref{thm:asy} and \ref{thm:randper}.
For the first proof, we need the following, always under the tightness assumptions of 
Theorem \ref{thm:compare}.

\begin{lem}\label{lem:Hp-Hq}
For $\ep > 0$, $0 < \delta \le 1/2$ and $x \in \interior T$,
$$
\bigl|\HH_p(x) - \HH_q(x)\bigr| \le 2\ep + M_{\ep}\,\varphi(\delta) +
\frac{2}{e\ln 2} \frac{1}{\delta}\, \bigl\| p(\cdot|x)-q(\cdot|x)\bigr\|_1\,.
$$
\end{lem}

\begin{proof} 
If $y \in N(x)$ is such that $\bigl|p(y|x)-q(y|x)\bigr| < \delta$, then Lemma \ref{lem:phi} yields
$$
\big|\varphi\bigl(p(y|x)\bigr)-\varphi\bigl(q(y|x)\bigr)\big| \le \varphi(\delta).
$$
In view of this,
$$
\bigl|\HH_p(x) - \HH_q(x)\bigr| \le \sum_{y \in N(x)}  \big|\varphi\bigl(p(y|x)\bigr)-\varphi\bigl(q(y|x)\bigr)\big|
= \text{Sum}_1 + \text{Sum}_2
$$
with 
$$
\begin{aligned}
\text{Sum}_1 
&= \sum_{ y\,:\, |p(y|x)-q(y|x)| < \delta } \big|\varphi\bigl(p(y|x)\bigr)-\varphi\bigl(q(y|x)\bigr)\big|\\
&\le \sum_{ y \in N(x) \setminus N_{\ep}(x)} \Bigl(\big|\varphi\bigl(p(y|x)\bigr)\big|
+ \big|\varphi\bigl(q(y|x)\bigr)\big|\Bigr)
+ \sum_{ y \in N_{\ep}(x)} \varphi(\delta)\\[5pt]
&< 2\ep + M_{\ep}\,\varphi(\delta)\,.
\end{aligned}
$$
Next,
$$
\begin{aligned}
\text{Sum}_2 
&= \sum_{ y\,:\, |p(y|x)-q(y|x)| \ge \delta } 
\big|\varphi\bigl(p(y|x)\bigr)-\varphi\bigl(q(y|x)\bigr)\big|\\
&\le \sum_{ y\,:\, |p(y|x)-q(y|x)| \ge \delta } \frac{2\max \varphi}{\delta} \,\delta\\
&\le \frac{2\max \varphi}{\delta}  \sum_{ y\,:\, |p(y|x)-q(y|x)| \ge \delta } \big|p(y|x)-q(y|x)\big|\\
&\le \frac{2\max \varphi}{\delta} \,\bigl\| p(\cdot|x)-q(\cdot|x)\bigr\|_1\,,
\end{aligned}
$$
completing the proof.
\end{proof}

\begin{proof}[\rm \textbf{Proof of Theorem \ref{thm:compare}}]
We start by observing that
\begin{equation}\label{eq:lsharp}
\sum_{x \in \interior T} \frac{\mu_{\Prob}(x)}{\ell(x)} = \frac{\ell_{\sharp}(\Prob)}{\ell(\Prob)}. 
\end{equation}
Using Lemma \ref{lem:entropy}, we write
$$
\begin{aligned}
\left| \frac{\HH(\Prob)}{\ell(\Prob)} - \frac{\HH(\Q)}{\ell(\Q)}\right|
&= \left| \sum_{x \in \interior T} \left(\mu_{\Prob}(x)\frac{\HH_p(x)}{\ell(x)}
-\mu_{\Q}(x)\frac{\HH_q(x)}{\ell(x)}\right)\right|\\
&\le 
\sum_{x \in \interior T} \frac{\mu_{\Prob}(x)}{\ell(x)}\,\bigl|\HH_p(x) - \HH_q(x)\bigr| \;+\; 
\left| \sum_{x \in \interior T}  \frac{\HH_q(x)}{\ell(x)}\,
\bigl(\mu_{\Prob}(x)- \mu_{\Q}(x)\bigr)\right|\\[4pt]
&= \hspace*{2.2cm}\text{Sum}_{I}\hspace*{1.75cm} + \hspace*{1.75cm} |\text{Sum}_{I\!\!I}|
\end{aligned}
$$
For the following sequence of estimates to bound $\text{Sum}_{I}\,$, in (1) we use 
Lemma \ref{lem:Hp-Hq},
in (2) equation \eqref{eq:lsharp} and Lemma \ref{lem:norm}(b), in (3)
the Cauchy-Schwarz inequality, and in (4) once more \eqref{eq:lsharp}, 
as well as Proposition \ref{pro:div},
$$
\begin{aligned}
 \text{Sum}_{I} 
&\stackrel{(1)}{\le} \sum_{x \in \interior T} \frac{\mu_{\Prob}(x)}{\ell(x)}\,\Bigl(2\ep + M_{\ep}\,\varphi(\delta)\Bigr)
\;+\; \sum_{x \in \interior T} \frac{\mu_{\Prob}(x)}{\ell(x)}\,\frac{2\max \varphi}{\delta}\,
                     \big\|p(\cdot|x) - q(\cdot|x)\big\|_1 \\
&\stackrel{(2)}{=} \frac{\ell_{\sharp}(\Prob)}{\ell(\Prob)}\Bigl( 2\ep + M_{\ep}\,\varphi(\delta)\Bigr)
 \;+\; \frac{C}{\delta}
\sum_{x \in \interior T} \sqrt{\frac{\mu_{\Prob}(x)}{\ell(x)}}\,
\sqrt{\frac{\mu_{\Prob}(x)}{\ell(x)}\,\DD_{p,q}(x)}\\
&\stackrel{(3)}{\le}
\frac{\ell_{\sharp}(\Prob)}{\ell(\Prob)}\Bigl( 2\ep + M_{\ep}\,\varphi(\delta)\Bigr)
 \;+\; \frac{C}{\delta}
\sqrt{ \sum_{x \in \interior T} \frac{\mu_{\Prob}(x)}{\ell(x)} }
\sqrt{ \sum_{x \in \interior T} \frac{\mu_{\Prob}(x)}{\ell(x)}\DD_{p,q}(x) }\\
&\stackrel{(4)}{=} \frac{\ell_{\sharp}(\Prob)}{\ell(\Prob)}\Bigl( 2\ep + M_{\ep}\,\varphi(\delta)\Bigr)
 \;+\; \frac{C}{\delta}\sqrt{ \frac{\ell_{\sharp}(\Prob)}{\ell(\Prob)} }
   \sqrt{ \frac{\DD(\Prob\|\Q)}{\ell(\Prob)} }.
\end{aligned}
$$
Next, by Lemma \ref{lem:norm}(a),
$$
\begin{aligned}
&\sum_{x \in \interior T}  \frac{\HH_q(x)}{\ell(x)}\,
\bigl(\mu_{\Prob}(x)- \mu_{\Q}(x)\bigr) \\
&\;
\le \sum_{x \,:\, \mu_{\Prob}(x) > \mu_{\Q}(x)} \!\!A\,\bigl(\mu_{\Prob}(x)- \mu_{\Q}(x)\bigr)
+ \sum_{x \,:\, \mu_{\Prob}(x) < \mu_{\Q}(x)} \!\!a\, \bigl(\mu_{\Prob}(x)- \mu_{\Q}(x)\bigr) 
= \frac{A-a}{2} \, \|\mu_{\Prob} - \mu_{\Q}\|_1 \,.
\end{aligned}
$$
Exchanging the roles of $\mu_{\Prob}$ and $\mu_{\Q}\,$, we get the same upper bound.
This proves that $|\text{Sum}_{I\!\!I}| \le \frac{A-a}{2} \|\mu_{\Prob} - \mu_{\Q}\|_1 \,$, 
concluding the proof.
\end{proof}

At this point we remark that one might want to refrain from deriving an
upper bound on $\textrm{Sum}_{I}$ that involves Pinsker's inequality, as in the proof
of Theorem \ref{thm:compare}: instead one could work directly with $\textrm{Sum}_{I}\,$.
However, recall from above that there are converses to Pinsker's inequality, so that one
cannot expect big improvements from such a change of proof strategy. 
We next turn to the lemma which allows us to use asymptotic constantness of
the length-normalised local entropies $\HH_q(x)/\ell(x)$ in Theorem \ref{thm:asy}.

\begin{proof}[\rm \textbf{Proof of Lemma \ref{lem:variant}}]
We need to modify the estimate of the term $|\text{Sum}_{I\!\!I}|$ in the proof of 
Theorem \ref{thm:compare}. Consider the sets 
$\textsf{Pos} = \{x \in \interior T : \mu_{\Prob}(x) > \mu_{\Q}(x)\}$
and $\textsf{Neg} = \{x \in \interior T : \mu_{\Prob}(x) > \mu_{\Q}(x)\}$, as well as
$\textsf{Pos}_S = \textsf{Pos} \cap \interior T^S$ and 
$\textsf{Neg}_S = \textsf{Neg} \cap \interior T^S$. We split
$$
\begin{aligned} 
 \sum_{x \in \interior T}  &\frac{\HH_q(x)}{\ell(x)}\,
\bigl(\mu_{\Prob}(x)- \mu_{\Q}(x)\bigr) \\
&\le \sum_{x \in \textsf{Pos}_S} A\,\bigl(\mu_{\Prob}(x)- \mu_{\Q}(x)\bigr)
+ \sum_{x \in \textsf{Neg}_S} a\,\bigl(\mu_{\Prob}(x)- \mu_{\Q}(x)\bigr)\\
&\quad + \sum_{x \in \textsf{Pos} \setminus \textsf{Pos}_S} A_S^*\,\bigl(\mu_{\Prob}(x)- \mu_{\Q}(x)\bigr)
+ \sum_{x \in \textsf{Neg} \setminus \textsf{Neg}_S} a_S^*\,\bigl(\mu_{\Prob}(x)- \mu_{\Q}(x)\bigr) \\
&= A_S^* \sum_{x \in \textsf{Pos}} \bigl(\mu_{\Prob}(x)- \mu_{\Q}(x)\bigr)
+ a_S^* \sum_{x \in \textsf{Neg}} \bigl(\mu_{\Prob}(x)- \mu_{\Q}(x)\bigr) \\
&\quad+ (A - A_S^*)\sum_{x \in \textsf{Pos}_S} \bigl(\mu_{\Prob}(x)- \mu_{\Q}(x)\bigr)
+ (a - a_S^*)\sum_{x \in \textsf{Neg}_S} \bigl(\mu_{\Prob}(x)- \mu_{\Q}(x)\bigr) \\
&\le  \frac{A_S^*-a_S^*}{2} \, \|\mu_{\Prob} - \mu_{\Q}\|_1
+ (A - A_S^*)\,\mu_{\Prob}(\textsf{Pos}_S) + (a_S^* -a) \mu_{\Q}(\textsf{Neg}_S)\,.
\end{aligned}
$$
Now, by the definition \eqref{eq:nu} of $\mu_{\Prob}\,$, 
$$
\mu_{\Prob}(\textsf{Pos}_S) = \frac{\ell(\Prob_S)}{\ell(\Prob)}\,\mu_{\Prob_S}\!(\textsf{Pos}_S) 
\le \frac{\ell(\Prob_S)}{\ell(\Prob)}, \quad\text{and analogously}\quad
\mu_{\Q}(\textsf{Neg}_S) \le \frac{\ell(\Q_S)}{\ell(\Q)}\,.
$$
This and an exchange of the roles of $\Prob$ and $\Q$ lead to the stated upper bound.
\end{proof}

At last, we come to the proofs of the two theorems of Section \ref{sec:infinite}.

\begin{proof}[\rm \textbf{Proof of Theorem \ref{thm:asy}}]
Write $\partial_{\infty}T$ for the space of ends of $T$. By assumption, 
$\Prob(\partial_{\infty} T) > 0$. Therefore
\begin{equation}\label{ref:to-infty}
\ell_{\sharp}(\Prob_n) \ge n \,\Prob(\partial_{\infty} T) \to \infty\,, \quad\text{as}\; n \to \infty\,.
\end{equation}
The same applies to $\Q$. 

Next, we use Lemma \ref{lem:variant}. We start with an arbitrary $\ep > 0$.  
By (ii) there is an index $k$ such that  $A^*_{S(k)} - a^*_{S(k)} < \ep$. 
We take $n > k$ and consider the sub-tree $T^{S(n)}$ with $\Prob_n$ and $\Q_n\,$. 
We note that $L = \ell_{\sharp}(\Prob_n)/ \ell(\Prob_n) \le c_2$ by assumption (i).
Next,  
$$
\delta_n := \left(\frac{\DD(\Prob_n\|\Q_n)}{\ell(\Prob_n)}\right)^{\! 1/4} \to 0 \quad 
\text{as }\;n \to \infty
$$ 
by assumption (i) and (iii) together with \eqref{ref:to-infty}.

Now the upper bound of Theorem \ref{thm:compare}, combined with the variant provided
by Lemma  \ref{lem:variant}, yields
$$
\begin{aligned}
\left|\frac{\HH(\Prob_n)}{\ell(\Prob_n)} 
      - \frac{\HH(\Q_n)}{\ell(\Q_n)}\right|
< \ &c_2 \Bigl( 2\ep + M_{\ep}\,\varphi(\delta_n) \Bigr)
\,+\, C\sqrt{c_2}\,\delta_n\\ 
&+ \ep 
\,+\,
 (A-a) \,\max \left\{ \frac{\ell(\Prob_{S(k)})}{\ell(\Prob_n)}\,,\, 
                     \frac{\ell(\Q_{S(k)})}{\ell(\Q_n)} \right\}\,.
\end{aligned}
$$
If we let $n \to \infty$, again using assumption (i) together with \eqref{ref:to-infty}
to control the last maximum, we see that
$$
\limsup_{n \to \infty}
\left|\frac{\HH(\Prob_n)}{\ell(\Prob_n)} 
      - \frac{\HH(\Q_n)}{\ell(\Q_n)}\right| \le (2 c_2+1)\, \ep\,,
$$ 
because all other terms tend to $0$.
We are left with showing that $\HH(\Q_n)/\ell(\Q_n) \to h$.
This is standard: with $\ep$ and $k$ as above,
$$
\begin{aligned}
\left|\frac{\HH(\Q_n)}{\ell(\Q_n)} - h\right|
&\le \sum_{x \in \interior T^{S(n)}} 
\left| \frac{\HH_q(x)}{\ell(x)} - h \right|\, \mu_{\Q_n}(x)\\
&<  \sum_{x \in \interior T^{S(k)}} (A+h)\, \mu_{\Q_n}(x) +\hspace*{-2.5mm}
 \sum_{x \in \interior T^{S(n)} \setminus T^{S(k)}} \hspace*{-2.5mm}\ep\, \mu_{\Q_n}(x) 
\le (A+h) \frac{\ell(\Q^{(k)})}{\ell(\Q_n)} + \ep. 
\end{aligned}
$$
This becomes smaller than $2\ep$ when $n$ is large enough. 
\end{proof}

And finally, we prove the result on random perturbations.

\begin{proof}[\rm \textbf{Proof of Theorem \ref{thm:randper}}]
First of all, we need to verify that the (random) local entropies $\HH_p(x)$ are 
uniformly summable. By tightness of $q(\cdot|x)$ and $q'(\cdot|x)$, there is $\wt M$
such that for each $x$ there is a set $\wt N(x) \subset N(x)$ with $|\wt N(x)| \le \wt M$
such that 
$$
\max \{ q(y|x), q'(y|x)\} \le 1/(2e) \quad \text{for all}\quad y \in 
N(x)\setminus \wt N(x)\,.
$$
Using Lemma \ref{lem:phi}(iv),
$$
\sum_{y \in N(x)\setminus \wt N(x)} \varphi\bigl(p(y|x)\bigr) 
\le 2 \sum_{y \in N(x)\setminus \wt N(x)} \Bigl(\varphi\bigl(q(y|x)\bigr) + \varphi\bigl(q'(y|x)\bigr)\Bigr) 
$$
As $x$ varies, the right hand sides are uniformly summable along with
$\HH_q(x)$ and $\HH_{q'}(x)$. 

At this point, with the same line of reasoning as in the proof of Theorem \ref{thm:asy},
what remains is to show that $\DD(\Prob_n\|\Q_n)/\ell(\Prob_n)\to 0$
in probability, resp. almost surely. For this purpose, we first take $x \in S(n)$,
and use convexity, see \cite[Thm. 2.7.2]{CoTh}:
$$
\begin{aligned}
\DD_{p,q}(x) &= \DD\Bigl((1-\dde_n)\,q(\cdot|x) + \dde_n\, q'(\cdot|x)\big\|(1-\dde_n)\,q(\cdot|x) 
+ \dde_n\, q(\cdot|x)\Bigr)\\
&\le \dde_n \,  \DD\bigl(q'(\cdot|x)\big\|q(\cdot|x)\bigr) \le \dde_n\,D\,.
\end{aligned}
$$
Now we combine the last inequality with Proposition \ref{pro:div} and the Definition 
\eqref{eq:nu} of $\mu_{\Prob_n}\,$,
$$
\frac{\DD(\Prob_n\|\Q_n)}{\ell(\Prob_n)} 
\le \Ex_{\mu_{\Prob_n}}\left( \frac{\dde_{|x|}\, D}{\ell(x)} \right) 
= \sum_{k=0}^n \sum_{x \in S(k)} \mu_{\Prob_n}(x) \frac{\dde_{k}\,D}{\ell(x)} 
= \frac{D}{\ell(\Prob_n)} \sum_{k=0}^n \dde_k  \le \frac{D}{c_1 n}\sum_{k=0}^n \dde_k\,.
$$
The last inequality follows from assumption (i). 
If $\Ex(\dde_n) \to 0$, resp. $\dde_n \to 0$ almost surely then
$$
\frac{D}{n}\sum_{k=0}^n \dde_k \to 0
$$
in probability, resp. almost surely. 
\end{proof}

\end{document}